\documentclass[conference,10pt,a4paper]{IEEEtran}

\usepackage{graphicx}
\usepackage{amsthm}
\usepackage{amssymb}
\newtheorem{lem}{Lemma}

\usepackage{subfig}
\usepackage{amsmath}
\usepackage{dsfont}
\usepackage{units}
\usepackage{hyperref}
\usepackage{abbrevs} 

\usepackage{tikz}
\usetikzlibrary{matrix}

\newabbrev\ISM{Slot/Matching (SM)}[SM]
\newabbrev\ACC{Access Number (ACC)}[ACC]

\makeatletter
\def\imod#1{\allowbreak\mkern10mu({\operator@font mod}\,\,#1)}
\makeatother

\def\hex#1{\ensuremath{\text{#1}_\text{h}}}

\newcommand\equationref[1]{Eqn.~(\ref{#1})}
\newcommand\figureref[1]{Fig.~(\ref{#1})}
\newcommand\sectionref[1]{Sec.~(\ref{#1})}
\newcommand\tableref[1]{Tab.~(\ref{#1})}

\begin{document}

\title{Reliable Reception of Wireless Metering Data with Protocol Coding}

\author{\IEEEauthorblockN{Rasmus Melchior Jacobsen\IEEEauthorrefmark{1}\IEEEauthorrefmark{2}}
\IEEEauthorblockA{\IEEEauthorrefmark{1}Kamstrup A/S\\
Denmark\\
rmj@kamstrup.dk}
\and
\IEEEauthorblockN{Petar Popovski\IEEEauthorrefmark{2}}
\IEEEauthorblockA{\IEEEauthorrefmark{2}Aalborg University, Dept. of Electronic Systems\\
Denmark\\
petarp@es.aau.dk}}

\maketitle

\begin{abstract}
Stationary collectors reading wireless,
battery powered smart meters,
often operate in harsh channel conditions to cut network installation cost to a minimum,
challenging the individual link to each meter.
The desired performance measure is reliable reception of at least some data from as many as possible meters, 
rather than increasing the fraction of received packets from one meter. 
As a first step for improving reliable reception,
and motivated by the recent revision of Wireless M-Bus,
we propose the use of a deterministic packet transmission 
interval to group packets from the same meter.
We derive the probability of falsely pairing packets from different senders in the simple case of no channel errors,
and show through simulation and data from an experimental deployment the probability of false pairing with channel errors.
The pairing is an essential step towards recovery of metering data from as many as possible meters under harsh channel conditions.
From the experiment we find that more than 15\% of all conducted pairings are between two erroneous packets,
which sets an upper bound on the number of additional meters that can be reliably recovered.
\end{abstract}

\section{Introduction}

Large-scale deployments of battery-powered
wireless smart meters for the heat, cooling, and water market,
puts unusual requirements on the infrastructure for efficient system operation.
Meters are mainly transmitter-only,
broadcasting devices,
and reliable reception cannot be achieved using feedback \cite{Massey85}.
The meters send a relatively small amount of data (temperature, volume, energy, etc.),
and a relevant reliability measure is the number of \emph{individual meters}, i.e. meters from which at least one packet has been received. 
In Europe, the dominant protocol for wireless transmission of metering data from battery powered meters,
is the Wireless M-Bus standard \cite{13757-4:2011} with around 20 million deployed meters.
Considering the constraints induced by the standard, such as 
power-limited operation in the $868\unit{MHz}$ ISM and the 
established protocol structure, the room for optimization is primarily at the receiver's side.

In such a setting, we extract information about the data in the packets by using the timing structure in the access protocol. 
The use of protocol actions to encode data has been termed \emph{protocol coding} \cite{trillingsgaard2013}.
Specifically we will use the inherent nature of the \emph{packet transmission time interval} which is known to contain information \cite{394647}.
The interval in Wireless M-Bus has a deterministic structure as a consequence of other design choices.
In its recent revision, much focus has gone into increasing the battery efficiency of the meter,
but also to allow for efficient operation on battery powered receivers.
In this process,
a procedure for accurate intervals between broadcast packet transmissions was introduced to allow for a battery-powered receiver to synchronize to the transmission scheme for a meter,
thereby allowing the receiver to sleep and only start a receive window right before the transmission from the meter.
In addition,
a meter can possibly transmit (repeat) identical data set for multiple times.
Even if the data is encrypted,
then the same encrypted payload can be transmitted multiple times (up to 300 seconds).
This is a feature that is used by many devices,
as it allows to skip the encryption process before each transmission,
thereby saving energy.
Meters that operate in this way are supplied by, for example, Kamstrup.

\begin{figure}
  \includegraphics[width=\columnwidth]{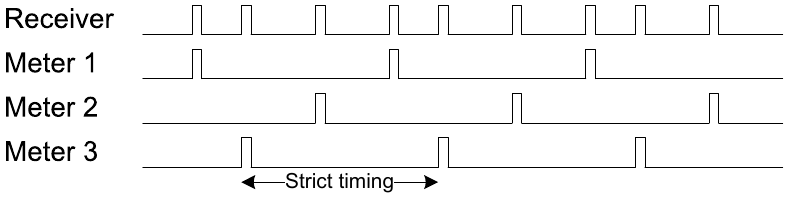}
  \caption{Packets with predictable transmission intervals.}
  \label{fig:overview}
  \vskip-0.7cm
\end{figure}

In this paper we utilize the timing structure in the transmission interval,
such that the receiver maps/groups the erroneously received packets from the same sender based on the packet transmission time.
An example of this is in \figureref{fig:overview},
where just by observing the arrival times on the receiver,
and knowing the transmission interval,
it is possible to readily group the packets arriving from distinct senders.
This grouping is a natural first step towards the recovery of meter data, as the receiver can apply 
various packet recovery algorithms to combine two or more erroneously received packets.
This is relevant under heavily error-prone channel conditions,
where packets experience high error rates,
if they arrive at all,
and where also the sender address is not reliable to pair packets.
This often occurs in deployments where reading reliability can be compromised due to network cost.


The optimal pairing can be achieved by joint decoding of packets based on their arrival times.
Instead, a feasible and practical,
progressive pair/no pair approach is proposed which we call the \ISM algorithm.
The \ISM algorithm judges upon every arrival at the receiver,
if a previous, erroneously received packet was registered,
for which the next packet is expected to arrive at the time of the current arrival.
The outcome can be \emph{pair} when the algorithm decides for the packets to originate from the same sender,
or alternatively \emph{no pair}.
In the event of pair between two erroneous packets,
the gain follows in a subsequent recovery process working across packets,
potentially restoring the otherwise unrecoverable data.

\begin{figure*}[ht!]
  \includegraphics[width=\textwidth]{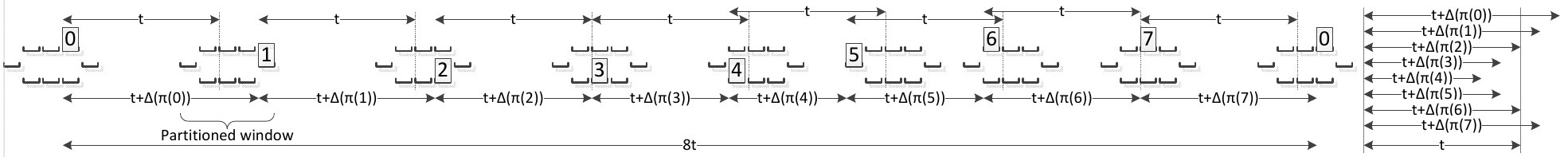}
  \caption{ACC cycle for $L=8$ transmissions from a meter.}
  \label{fig:slotsequence}
  \vskip-0.7cm
\end{figure*}

The system model is in the next section. 
The proposed packet pairing algorithm is described in \sectionref{sec:pp},
followed by an analysis of the probability of false detection in \sectionref{sec:fa},
being the case where the algorithm outcome is pair,
when the packets do not originate from the same meter.
Data from an experiment operating under severe channel conditions is given in \sectionref{sec:deployment},
followed by a discussion on the implementation feasibility in \sectionref{sec:feasability}.
Lastly, the conclusion is in \sectionref{sec:conclusion}.
 
\section{Background and System Model}
\label{sec:sm}


For a single meter,
the interval from packet $i$ to the next packet from the same meter $(i+1)$ depends on a specific field in the $i$th packet:
The \ACC field $x_i$.
The field is 1 byte, and for each transmission from a meter,
the access number is incremented with overflow in modulo $L=256$; that is $x_{i+j}=(x_i+j)\imod L$.
The interval is then $t+\Delta(\pi(x_i))$,
where $t$ is the average transmission interval and $\pi(x_i) = \left|x_i-\frac{L}{2}\right|$ is a jitter index,
and where $\Delta(\cdot)$ specifies the time offset relative to the nominal interval for a given jitter.

The $L$ possible values of $x_i$ each represent a \emph{slot} where the next packet $(i+1)$ with \ACC $x_{i+1}$ should be transmitted.
Note that each possible slot is \emph{not} uniquely separated in time,
but separated in time \emph{and} by the code given in the ACC $x_i$.
A resulting full \ACC cycle for a reduced value of $L=8$ is shown in \figureref{fig:slotsequence}.
From the figure it is clear that the interval varies per transmission,
following the cyclic behaviour in \figureref{fig:accslots} with $\frac{L}{2}+1$ jitter values.
The $L$ possible slots in where packet $(i+1)$ can arrive if $x_i$ is considered unreliable make up a (timely partitioned) \emph{window}.
Each transmitted \ACC $x_i$ from a meter has the received counterpart $y_i$,
where each bit in $y_i$ is randomly and independently in error with probability $\epsilon$.
In addition, a packet from a meter can also be dropped without being received at all (erasure) with probability $p$.

\begin{figure}
  \centering
  \includegraphics[width=\columnwidth]{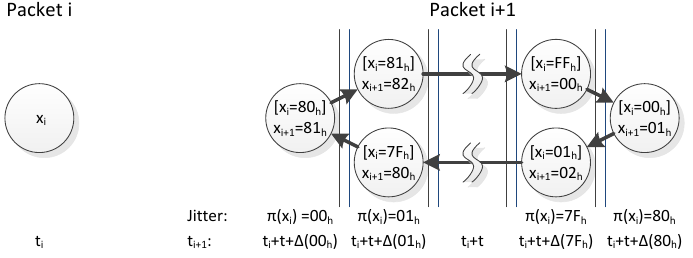}
  \caption{The possible slots for packet $(i+1)$ depending on $x_i$. The values of $x_i$ in the square brackets are the requirement for the packet $(i+1)$ to arrive in the particular slot.}
  \label{fig:accslots}
  \vskip-0.7cm
\end{figure}

For a base packet $i$ transmitted at time $t_i$ with \ACC $x_i$,
the nominal interval to a future packet $(i+j)$ from the same meter is $t_{x_i}^{nom}(j)=\sum_{j'=0}^{j-1}\left[t+\Delta(\pi(x_{i+j'}))\right]$.
This nominal transmission start time falls within the associated slot time boundaries,
where we define $\theta_{x_i}(j)$ to be the time relating to \ACC $x_i$ such that the slot for window $j$ starts at time $t_i+t_{x_i}^{nom}(j)-\theta_{x_i}(j)$.
The slot has width $\tau_{x_i}(j)$.
For Wireless M-Bus,
$\theta_{x_i}(j)$ and $\tau_{x_i}(j)$ are:
\begin{align*}
\theta_{x_i}(j)&=t_{x_i}^{nom}(j)\nu_a+\gamma_a,\\
\tau_{x_i}(j)&=t_{x_i}^{nom}(j)(\nu_a+\nu_b) + \gamma_a + \gamma_b,
\end{align*}
where $\nu_a$ describes a cumulative jitter and $\gamma_b$ a non-cumulative jitter.
$\nu_a=30\unit{ppm}$, $\nu_b=110\unit{ppm}$ under normal operating conditions.
The non-cumulative jitter is $1\unit{ms}$ per packet so $\gamma_a=\gamma_b=2\unit{ms}$ when considering any two packets.

\section{Packet Pairing}
\label{sec:pp}
The objective of the packet pairing is to provide a first step towards packet recovery on otherwise unrecoverable erroneous receptions where packets are combined across transmissions from the same meter with the same metering data.
The objective of the packet pairing is to relate an unreliable \ACC observation $y_i$,
with a, to be transmitted later, \ACC observation $y_\alpha$ which w.h.p. comes from the same meter.
Any other known parts from the received packets could also be used for the pairing in addition to the \ACC,
for example the sender address,
which would only increase the reliability of a correct pairing.
Any such addition is however not included in the analysis.
Remarkably, even though the \ACC field is only 8 bits,
it is actually possible,
as we shall see later,
to distinguish many more than 256 devices with high reliability.

While the \ACC field is considered in this work because the transmission interval between two packets depends on this field,
the work can be extended to many different flavors depending on the analyzed protocol.
An example is fixed interval, which corresponds almost to a TDMA schedule.
The trouble with such a setup is that it requires two-way communication to align transmissions from all devices in time.
What the \ACC jitter approach gives,
is that two-way ``scheduling'' is not needed,
as if two devices at one time have a colliding transmission,
they will w.h.p. be separated in time at the next transmission.
In fact, for two meters to continue having colliding transmission,
they have to have the exact same \ACC at the same transmission time,
which is an unlikely event.

\subsection{Error Events}
A transmitted \ACC from a meter is incremented for each transmission,
and its value affects the transmission interval.
This introduces a dependency on the error across any two packets $i$ and $(i+j)$ from the same meter.
A packet is said to be in error if it fails to comply with some error detection mechanism such as CRC,
and the \ACC part of the packet should be considered unreliable.
Take one unreliable \ACC observation $y_i$ with the underlying actually transmitted \ACC $x_i$.
Let $H(a,b)$ denote the Hamming distance between $a$ and $b$,
then with the independent bit error model,
we have $H(x_i,y_i)=0$ with probability $(1-\epsilon)^8$.
$H(x_i,y_i)=1$ with probability $8(1-\epsilon)^7\epsilon$, etc.
This uncertainty has a direct impact on when to expect the next packet from the same meter.
To capture this uncertainty,
we will introduce the notion of a \emph{virtual slot},
representing a possible slot in the future for the next transmission from the same meter.
For one erroneous \ACC $y_i$, a maximum of $L$ virtual slots can exist.
The probability of the next transmission $(i+1)$ from the same meter to fall within a specific virtual slot,
is a function of the number of bit errors incurred by picking the slot.
Specifically, from one base packet $i$ the set of virtual slots which imposes $b$ bit errors in $y_i$ are:
$\mathcal{X}_b(j)=\{\xi:H(y_i,\xi-j)=b\}$,
where $\mathcal{X}_b(j)$ is the set of virtual slots with $b$ ``known'' bit errors after $j$ steps.
Clearly $|\mathcal{X}_b(j)|={8\choose b}$.
For relevant low values of $\epsilon$ this means that there is one favorite virtual slot at step $j$,
namely $\xi=y_i+j$.
There are 8 less likely virtual slots, one for each possible single bit error in $y_i$, etc.

From this there are the following error events:
1) Bit error in $y_i$ introducing uncertainty about the arrival time of the next packet,
2) bit error in $y_{i+j}$ which together with the virtual slot specifies how likely the packets are to originate from the same meter,
3) false detection, where a packet arrives from a different meter in a virtual slot setup for $y_i$,
in a way where the pairing decides for them to come from the same meter.
The two first events are used next,
and the false detection event is analyzed in \sectionref{sec:fa}.

\subsection{The Slot/Matching Algorithm}
The \ISM algorithm is the basis in the pairing of packets from the same meter presented in this paper.
While there are many different ways to conduct the pairing, the \ISM algorithm is simple and puts a relaxing set of requirements to the receiver.
The algorithm operates in three steps, all triggered on the arrival of a packet $i$.
The steps depends on a storage structure comprising of virtual slots created for erroneously received packets.
Specifically, a virtual slot in the \ISM algorithm specifies:
1) Time boundaries (start and end time),
2) Reference to the erroneous base packet who initiated the creation of the virtual slot,
3) The number of bits in error in $y_i$ being the value of $b$ in $\mathcal{X}_b$ in which the slot belongs,
4) Expected \ACC for packet to arrive in the slot,
5) Step count $j$, specifying the number of intervals since the base packet.

Virtual slots are created when an erroneous packet is received.
The number of virtual slots created depends on the threshold policy used.
A simple policy is:
If the total Hamming distance on two packets $j$ steps apart:
\begin{equation}
D=H(\xi_\alpha-j,y_i)+H(\xi_\alpha,y_\alpha),
\label{eqn:D}
\end{equation}
is less than or equal $M$ bits,
then pair the packets.
In $D$, $i$ is the base packet and $\xi_\alpha$ is the expected \ACC registered for the virtual slot.
This gives, if $M=0$,
then only one virtual slot is created per arriving erroneous packet $i$,
namely the slot uniquely given by $y_i$ which is $\xi=y_i+j$.
Any other possible slot would have caused us to allow a bit error in $y_i$.
If $M=1$ is the threshold value, a total of 9 virtual slots are created, etc.

\begin{table}
\centering\small
\begin{tabular}{|lllll|}\hline
$\xi$ & $H(y_i,\xi-j)$ & Step & Start time & Duration \\\hline
$\hex{41}$ & 0 & 1 & $t_i+t^{nom}_{\hex{40}}(1)-\theta_{\hex{40}}(1)$ & $\tau_{\hex{40}}(1)$ \\
$\hex{42}$ & 1 & 1 & $t_i+t^{nom}_{\hex{41}}(1)-\theta_{\hex{41}}(1)$ & $\tau_{\hex{41}}(1)$ \\
$\hex{43}$ & 1 & 1 & $t_i+t^{nom}_{\hex{42}}(1)-\theta_{\hex{42}}(1)$ & $\tau_{\hex{42}}(1)$ \\
$\hex{45}$ & 1 & 1 & $t_i+t^{nom}_{\hex{44}}(1)-\theta_{\hex{44}}(1)$ & $\tau_{\hex{44}}(1)$ \\
$\hex{49}$ & 1 & 1 & $t_i+t^{nom}_{\hex{48}}(1)-\theta_{\hex{48}}(1)$ & $\tau_{\hex{48}}(1)$ \\
$\hex{51}$ & 1 & 1 & $t_i+t^{nom}_{\hex{50}}(1)-\theta_{\hex{50}}(1)$ & $\tau_{\hex{50}}(1)$ \\
$\hex{61}$ & 1 & 1 & $t_i+t^{nom}_{\hex{60}}(1)-\theta_{\hex{60}}(1)$ & $\tau_{\hex{60}}(1)$ \\
$\hex{01}$ & 1 & 1 & $t_i+t^{nom}_{\hex{00}}(1)-\theta_{\hex{00}}(1)$ & $\tau_{\hex{00}}(1)$ \\
$\hex{C1}$ & 1 & 1 & $t_i+t^{nom}_{\hex{C0}}(1)-\theta_{\hex{C0}}(1)$ & $\tau_{\hex{C0}}(1)$ \\\hline
\end{tabular}
\caption{Virtual slots created upon receiving erroneous packet $i$ with ACC $y_i=\hex{40}$.}
\label{tab:vslots}
\vskip-0.7cm
\end{table}

An example of the 9 virtual slots created for $M=1$ when an erroneous packet is received with \ACC $y_i=\hex{40}$ is in \tableref{tab:vslots}.
In this case $\mathcal{X}_0(1)=\{\hex{41}\}$ and $\mathcal{X}_1(1)=\{\hex{42},\hex{43},\hex{45},\hex{49},\hex{51},\hex{61},\hex{01},\hex{C1}\}$.
$H(y_i,\xi-j)$ is the Hamming distance between the received \ACC $y_i$ and the \ACC registered for the virtual slot minus the step count (\ACC increment).
Note from the table that there can be multiple slots $\xi$ for the same base packet describing the same time interval.
In the example this is true for $\xi=\hex{41}$ and $\xi=\hex{C1}$ as their jitter index are equal: $\pi\left(\hex{40}\right)=\pi\left(\hex{C0}\right)$.
One could think that for the virtual slots representing the same time,
but with different $\xi$,
that it is sufficient to only store the virtual slots with the shortest Hamming distance.
However, if the step count $j$ is incremented in the event where the packet is not found in the $j$th window,
then the time for the virtual slots diverges:
$\pi\left(\hex{40}\right)=\pi\left(\hex{C0}\right)$ but $\pi\left(\hex{41}\right)\neq\pi\left(\hex{C1}\right)$.

The algorithm upon receiving a packet $i$ is as summarized:
\begin{enumerate}
  \item If there are any virtual slots registered where the arrival time for packet $i$ is included in the slot boundaries, then for each virtual slot:
  \begin{itemize}
    \item Determine if candidate is sufficiently good according to \equationref{eqn:D},
    and if so, pair and remove all virtual slots $\mathcal{X}_b(\cdot)$ created by the same packet as now matches $i$.
  \end{itemize}
  \item For virtual slots where the slot boundary end time has passed,
  recompute slot start time, duration, increment $\xi$ and increment step count.
  If step count is too large (timeout), remove the virtual slot.
  \item If packet $i$ is erroneous, then add new virtual slots.
\end{enumerate}

A proper timeout on the number of times a virtual slot is recomputed depends on:
1) If the pairing should be used for packet recovery,
the timeout should not be larger than the expected number of times the same packet payload may be transmitted from the same meter,
and 2) by increasing the timeout, the receiver will require more memory as virtual slots can remain longer in the system (more on this later).

\section{False Detection Probability}
\label{sec:fa}
The \ISM algorithm makes an instantaneous decision about whether a pairing is valid or not upon the arrival of a packet,
and it is therefore relevant to find the probability of false detection:
The probability for a pairing to be established between two packets not originating from the same meter.
The analysis will make the following assumptions:
\begin{itemize}
  \item Collisions do not occur and transmissions are considered instantaneous.
  \item The pairing decision will only use the 8 \ACC bits,
  and ignore any improvement which can be added to further reduce the probability of false detection.
  \item Infinite population allowing to model the transmission rate for all meters as Poisson with load $\lambda=\frac{n}{t}$, $n$ being the number of meters in range of the receiver.
  \item Equal erasure probability $p$ for all meters.
  \item No drift in the transmission interval on a meter,
  hence the lead time is known and $\theta_{x_i}(j)$ for all meters.
\end{itemize}
The first assumption may seem controversial and reflects the case where a receiver,
in parallel,
is able to simultaneously receive multiple packets,
and where the errors incurred by this operation can be modelled with $\epsilon$ and $p$.
Any other receiver where only one packet can be received at a time,
will have a false detection probability lower than what we find in this section,
as the number of registered arrivals in the critical periods are fewer.
The assumption on Poisson arrivals reflect a normal operating system with no memory,
where all arrivals except the one we expect are random and untracked.
For the \ISM algorithm this is a valid approximation when there are no adversarial behaving meters,
but for any more advanced algorithm considering all meters,
clearly the full state space of all meters should be considered as a whole.

We will derive the false detection probability for the case where the bit error rate $\epsilon=0$,
and where the erasure probability is $p=0$.
First with a motivating example for $M=0$ and then for any $M$.
Through simulation the false detection probabilities for various values of $p$ and $\epsilon$ are found.

\subsection{False Detection Probability with $M=0$, $p=0$, and $\epsilon=0$}
When $M=0$, only one virtual slot is created per erroneously received base packet,
namely the slot $\xi$ for which $H(y_i,\xi-1)=0$ (or equivalently $\xi=y_i+1$).
Given $\epsilon=0$ and $p=0$,
we know that the desired packet \emph{will} arrive.
Also we know that the time from the slot start boundary to the arrival of the desired (true) packet from the same meter is $\theta_{\xi-1}(1)=\theta_{y_i}(1)$.
Abbreviate this time $\sigma$,
and the probability of false detection is completely given by the probability of a packet arriving during $\sigma$ having the exact same \ACC $\xi$ as being expected for the slot.
Assume uniform distribution of the received \ACC{}s for all other meters which can arrive during $\sigma$,
then the probability of an arriving packet to have exactly the same \ACC as the true packet $(i+1)$ is $\frac{1}{L}$.
It must be the exact same, otherwise a bit error is acceptable which is not the case for $M=0$.
For $k$ arriving packets during $\sigma$,
the probability of \emph{any} of them picking the same \ACC is $1-\left(1-\frac{1}{L}\right)^k$.
Summing the probability for $k>0$ Poisson arrivals where any of the $k$ arrivals have the problematic \ACC and the false detection probability for $M=0$ follows as:
\begin{align*}
q_0&=\sum_{k=1}^\infty(\lambda\sigma)^k\frac{e^{-\lambda\sigma}}{k!}\left(1-\left(1-\frac{1}{L}\right)^k\right)\\
&=1-e^{-\lambda\sigma}\left[1+\sum_{k=1}^\infty\frac{\left(\lambda\sigma\left(1-\frac{1}{L}\right)\right)^k}{k!}\right]\\
&=1-e^{-\frac{\lambda\sigma}{L}}.
\end{align*}

\begin{figure}
  \centering
  \subfloat[$\mathcal{A}=\{\{\hex{61}\},\{\hex{51}\},\{\hex{49}\},\{\hex{45}\},\{\hex{43}\},\{\hex{42}\}\}$, $\mathcal{B}=\{\hex{41},\hex{C1}\}$.]{\label{fig:vslots40h}\includegraphics[width=\columnwidth]{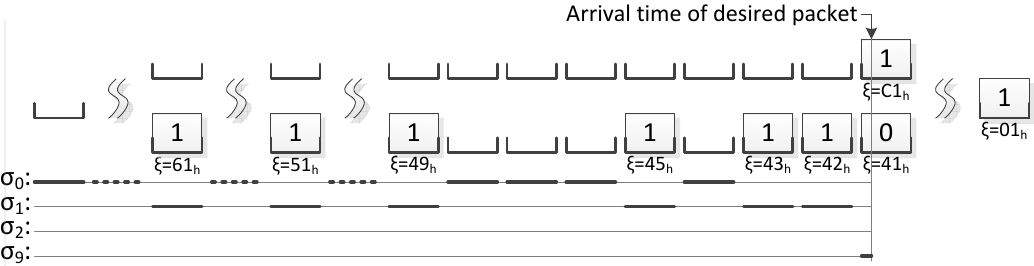}}\\
  \subfloat[$\mathcal{A}=\{\{\hex{61},\hex{A1}\},\{\hex{31}\},\{\hex{29}\},\{\hex{25}\},\{\hex{23}\},\{\hex{22}\}\}$, $\mathcal{B}=\{\hex{21}\}$.]{\label{fig:vslots20h}\includegraphics[width=\columnwidth]{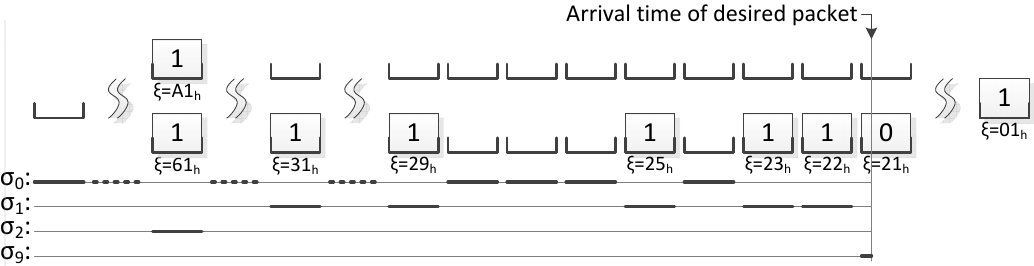}}\\
  \caption{Virtual slots for $M=1$ when $y_i=\hex{40}$ and $y_i=\hex{20}$.}
  \label{fig:vslots}
  \vskip-0.7cm
\end{figure}

\subsection{General False Detection Probability with $p=0$ and $\epsilon=0$}
Now the sum of bits in error must be less than or equal $M$.
Let $\mathcal{X}=\cup_{b=0}^M\mathcal{X}_b(1)$.
Only the virtual slots in $\mathcal{X}$ located \emph{before} the desired transmission is relevant when $p=0$ and $\epsilon=0$.
Define the sets:
\begin{align*}
\mathcal{A}&=\left\{a\in\mathcal{A}:a=\{\xi\in\mathcal{X}:\pi(\xi-1)=s\}\forall s=0,\ldots,\pi(y_i)-1\right\}\\
\mathcal{B}&=\left\{\xi\in\mathcal{X}:\pi(\xi-1)=\pi(y_i)\right\}\\
\mathcal{C}&=\mathcal{A}\cup\{\mathcal{B}\},
\end{align*}
where each element in the sets corresponds to a ``timebin'' which can contain multiple virtual slots.
$\mathcal{A}$ contains the timebins of full virtual slots located before the desired transmission,
and $\mathcal{B}$ is similar but for the virtual slots overlapping in time with the desired transmission.
Examples of $\mathcal{A}$ and $\mathcal{B}$ are in \figureref{fig:vslots}.
Now define $U(\xi)=\{u:H(\xi,u)\leq M-H(y_i,\xi-1)\}$ to be the surjective function defining the set of possible \ACC combinations for slot $\xi$,
such that the threshold policy of $M$ bits is satisfied.
Then the total number of \ACC combinations allowed for a timebin in $\mathcal{C}$ for a false detection is:
\[ d(c)=\left|\cup_{\xi\in c}U(\xi)\right|. \]
In the example in \figureref{fig:vslots40h} $d(a)=1$ for all $a\in\mathcal{A}$,
and $d(\mathcal{B})=9$ as the ``error contribution'' from the virtual slot $\xi=\hex{C1}$ is implicitly included in the one bit error combinations of $\xi=\hex{41}$.
We will use the number of combinations to divide the slot time before the arrival into parts,
depending on the possible number of error combinations.
Let the part allowing for $\beta$ \ACC combinations have duration:
\begin{equation}
\sigma_\beta=\sum_{a\in\mathcal{A}:d(a)=\beta}\tau'_a+\delta[d(\mathcal{B})-\beta]\theta_{y_i}(1),
\label{eqn:sigmabeta}
\end{equation}
where $\tau'_a$ is the duration of timebin $a$,
and where $\delta[\cdot]$ is the Kronecker delta ensuring that the duration of the last timebin is included in the duration allowing $d(\mathcal{B})$ false \ACC combinations.
There is a maximum of $L$ parts, where most of them have a duration of zero if the number of error combinations is not defined,
hence $\sum_{\beta=0}^L\sigma_\beta$ is the total duration of timebins before the expected arrival.

\begin{figure*}[ht]
  \centering
  \subfloat[$M=0$ for straight lines, otherwise $M=1$.]{\label{fig:falsealarm}\includegraphics[width=0.32\textwidth]{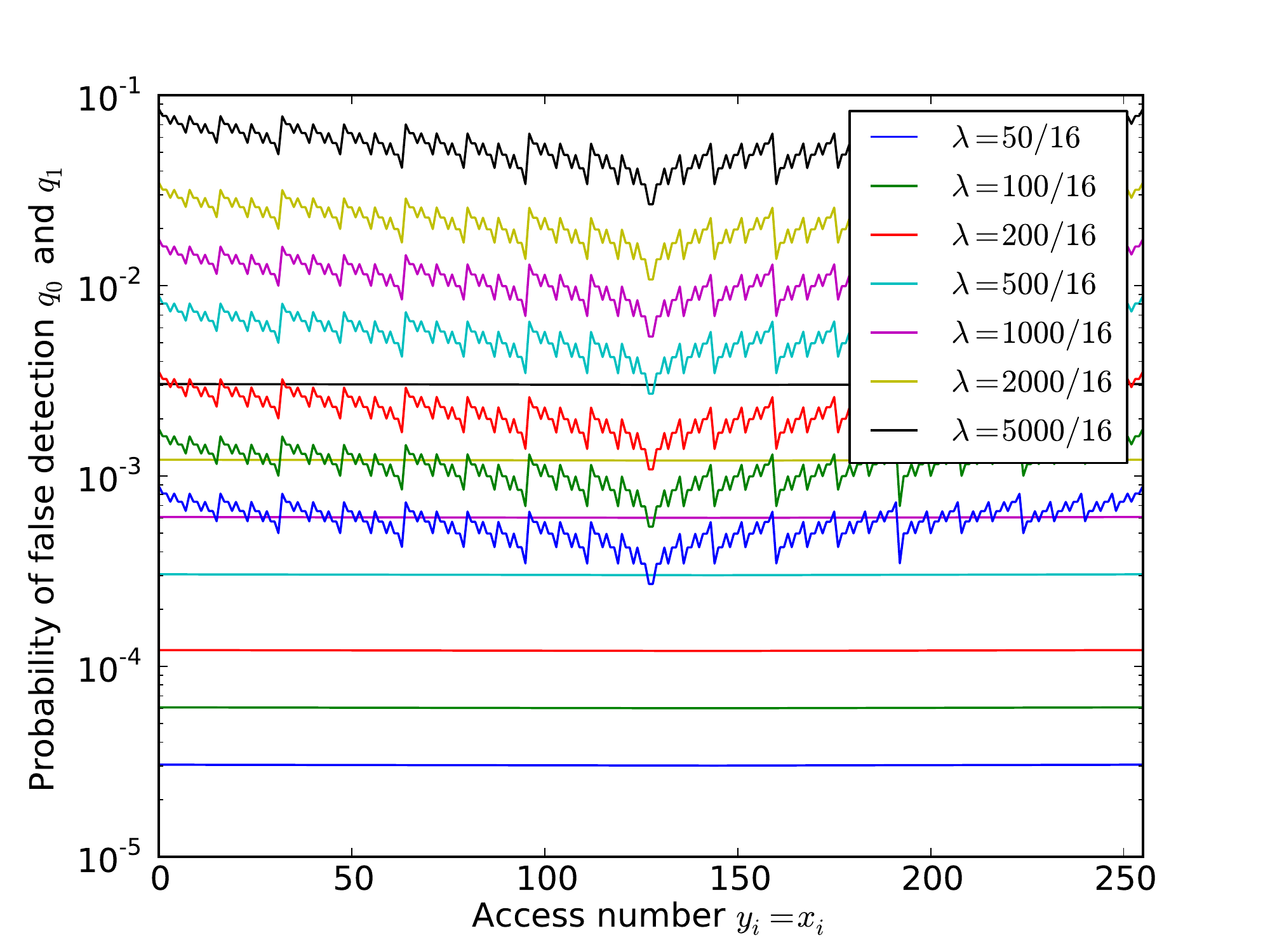}}
  \quad\subfloat[Simulated, $M=0$.]{\label{fig:erasuresweepM0}\includegraphics[width=0.32\textwidth]{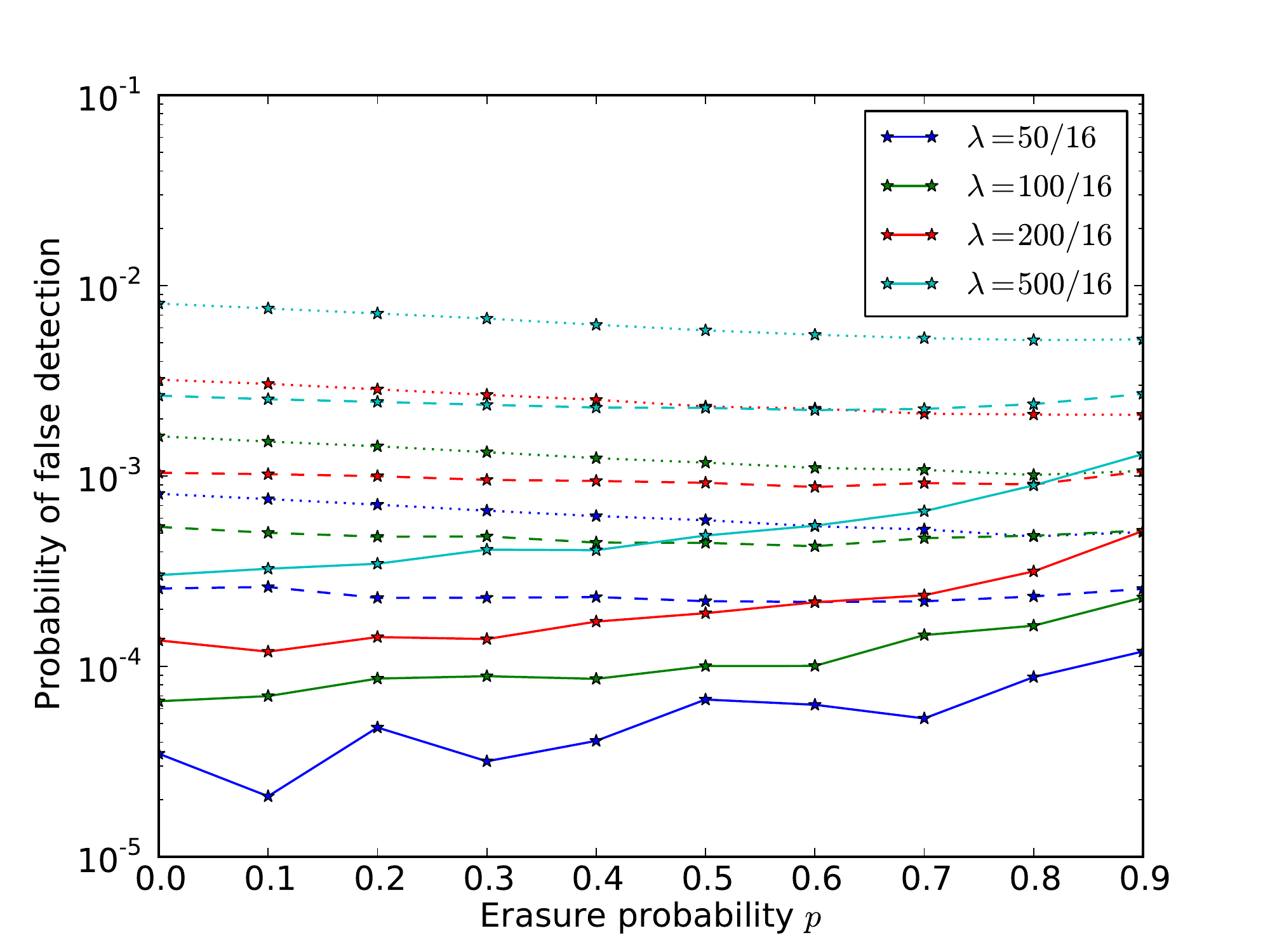}}
  \quad\subfloat[Simulated, $M=1$.]{\label{fig:erasuresweepM1}\includegraphics[width=0.32\textwidth]{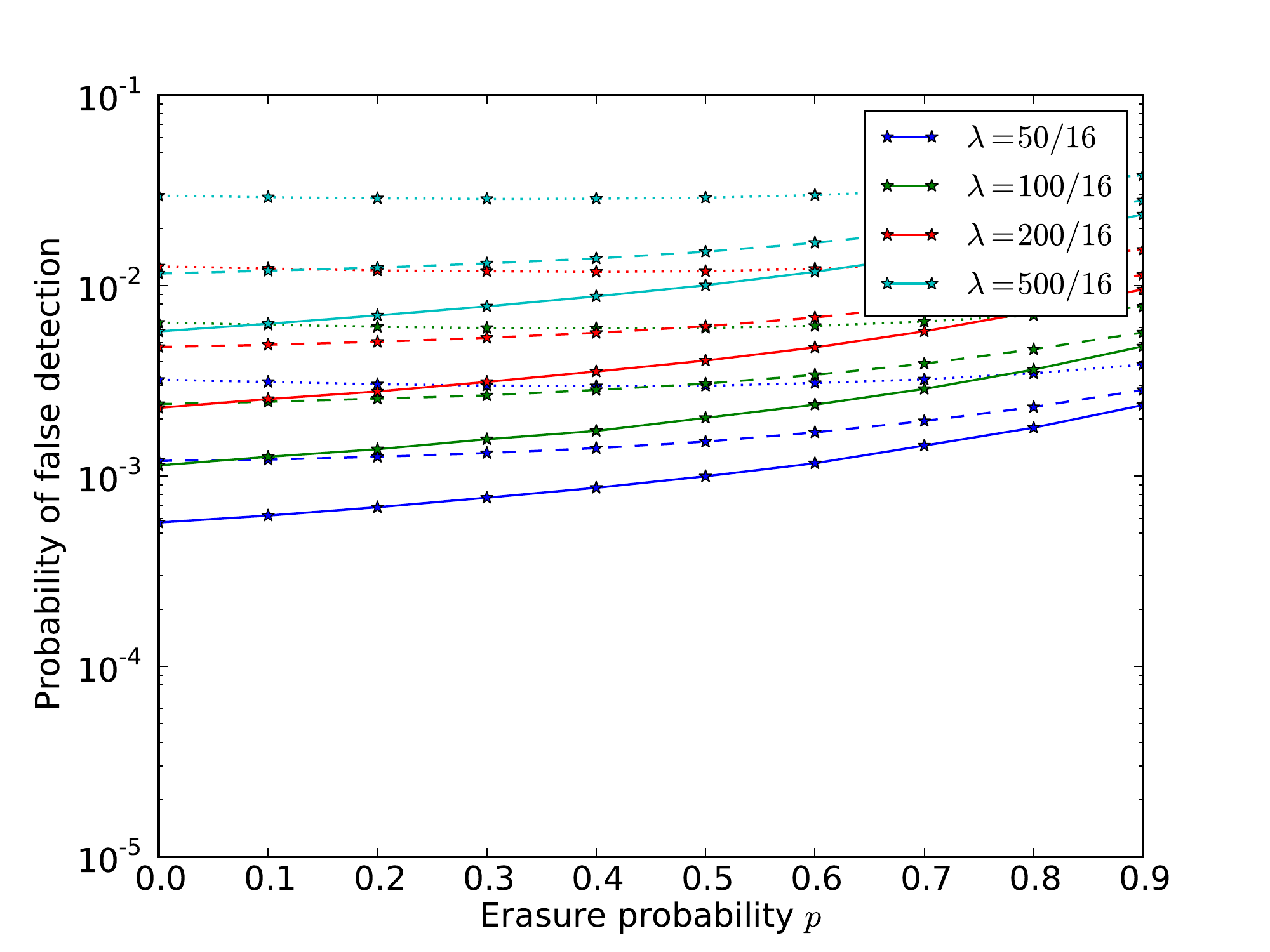}}
  \caption{Probability of false detection. $\epsilon=0$ for solid lines, $\epsilon=\frac{1}{32}$ for dashed lines, and $\epsilon=\frac{2}{32}$ for dotted lines.}
  \label{fig:erasuresweep}
  \vskip-0.7cm
\end{figure*}

\begin{lem}
The probability of false detection is:
\[ q_M=1-e^{-\frac{\lambda}{L}\sum_{\beta=1}^L\beta\sigma_\beta}, \]
when $\epsilon=0$ and $p=0$, with $\sigma_b$ defined as in \equationref{eqn:sigmabeta}.
\end{lem}

\begin{proof}
Let $K_\beta$ be the r.v. describing number of arrivals within the part with duration $\sigma_\beta$ allowing $\beta$ erroneous \ACC combinations,
and let $F_\beta$, $0\leq F_\beta\leq K_\beta$, be the r.v. describing the number of arriving, erroneous \ACC combinations.
From the definition that all events span the entire probability space:
\begin{align*}
&\sum_{k=0}^\infty\sum_{k_1+\cdots+k_L=k}\sum_{f_1=0}^{k_1}\cdots\sum_{f_L=0}^{k_L}\\
&\quad P(K_0=k_0,F_0=f_0,\ldots,K_L=k_L,F_L=f_L)=1,
\end{align*}
that is, there can be any number of arrivals within $\sum_{\beta=0}^L\sigma_\beta$ and they can distribute in any number of ways in the $L$ possible parts,
and in any part, there can be as many in error as there arrives.
The probability of false detection is given by the events where $\sum_{\beta=0}^Lf_\beta>0$.
Alternatively, and using that the events over time are independent observations:
\[ q_M=1-\sum_{k=0}^\infty\sum_{k_1+\cdots+k_L=k}\prod_{\beta=0}^LP(K_\beta=k_\beta,F_\beta=0). \]
The arrivals during $\sigma_\beta$ are random and Poisson with rate $\lambda\sigma_\beta$.
Continuing with $(1-q_M)$:
\begin{align*}
1-q_M&=\sum_{k=0}^\infty\sum_{k_1+\cdots+k_L=k}\prod_{\beta=0}^L(\lambda\sigma_\beta)^{k_\beta}\frac{e^{-\lambda\sigma_\beta}}{k_\beta!}\left(1-\frac{\beta}{L}\right)^{k_\beta}\\
&=\sum_{k=0}^\infty e^{-\lambda\sum_{\beta=0}^L\sigma_\beta}\sum_{k_1+\cdots+k_L=k}\prod_{\beta=0}^L\frac{\left(\lambda\sigma_\beta\left(1-\frac{\beta}{L}\right)\right)^{k_\beta}}{k_\beta!}\\
&=\sum_{k=0}^\infty\frac{e^{-\lambda\sum_{\beta=0}^L\sigma_\beta}}{k!}\left(\sum_{\beta=0}^L\lambda\sigma_\beta\left(1-\frac{\beta}{L}\right)\right)^k
\end{align*}
where the last equality follows from the multinomial theorem.
Continuing with the exponential power series, and inserting:
\begin{align*}
q_M&=1-e^{-\lambda\sum_{\beta=1}^L\sigma_\beta}e^{\sum_{\beta=1}^L\lambda\sigma_\beta\left(1-\frac{\beta}{L}\right)},
\end{align*}
which concludes the proof.
\end{proof}

The probability of false detection for the full range of possible \ACC{}s is plotted in \figureref{fig:falsealarm} for meters having an average transmission interval of 16 seconds.
In the figure the straight lines are for $M=0$,
where for $M=1$ the probability depends on the number of slots (and their duration) before the correct arrival which again depends on the \ACC.
It can be seen that the increase of $M$ from $M=0$ to $M=1$ increases the possibility of false detection with a little more than a decade.

Relating back to the statement in \sectionref{sec:pp} about reliably deciding to pair packets from the same sender for many more than 256 devices,
one can find from \figureref{fig:falsealarm} that the receiver can distinguish around 2000 devices,
and only make a false decision in less than $0.1\%$ of the pairings in the case where no bit errors are tolerated.

\subsection{Simulated False Detection for Various $\epsilon$ and $p$}
We will now consider simulated experiments in the same setting as for the analytic work,
but with nonzero bit error probability $\epsilon$ and erasure probability $p$.
The synchronization word in Wireless M-Bus is 32 bits,
and typical receivers, for example the cc1101 allows up to two bit errors in this sequence for the start of packet event to be detected.
We therefore run the experiments for values of $\epsilon=0,\frac{1}{32},\frac{2}{32}$.
The simulated results are in \figureref{fig:erasuresweepM0} and \figureref{fig:erasuresweepM1} averaged over the possible \ACC values.
It is clear that the simulated results matches the analytical result for $\epsilon=0$ and $p=0$.
Also it is interesting to see that the effect of $p$ is not as significant as the effect of $\epsilon$ when it comes to affecting the probability of false detection.

The results indicate that the \ISM algorithm should be parameterized with a low value of $M=0$ for correct operation,
and for relevant values of nonzero $\epsilon$ and $n=200$ meters a false detection will be made in the region around $0.1\%-1\%$ no matter the erasure probability $p$.

\section{Deployment Evaluation}
\label{sec:deployment}
We have tested the \ISM algorithm on a deployed receiver operating in harsh channel conditions.
The results are in \tableref{tab:experiment},
where the challenging channel conditions are clear from the number of erroneously received packets (detected by CRC).
The meters are Kamstrup Multical 402 where each transmitted packet with updated metering data is repeated six times every (on average) 16 second.
The experiment has run for 628 minutes,
corresponding to $\approx 2350$ transmitted packets for each meter within range of the receiver.
Since the transmissions are from actual meters,
there is no way to explicitly track the event of an erroneous pairing.
Instead we run the \ISM algorithm,
and for each pairing event,
we compare the two packets according to other fields within the packet,
and from this decide for the pairing validity.

The \ISM algorithm is run with $M=0$ and all arrivals,
no matter if the arrival is erroneous or not,
triggers the setup of a virtual slot.
This allows us to separate the pairing capability between a \underline{C}orrect/\underline{E}rroneous base packets and a \underline{C}orrect/\underline{E}rroneous arrival in a virtual slot.
For error correcting algorithms running on top of the \ISM algorithm,
the interesting measure is that for E$\rightarrow$E,
as the other three combinations immediately allow for receiving the payload data from at least one of the two packets in the pairing.
If we only consider the first step, $\approx15\%$ of all pairings are E$\rightarrow$E pairings (the number increases to $\approx18\%$ if two steps are considered).
Of these pairs it is likely that not all can be readily recovered,
but it is clear that a potential packet recovery gain exists.
Also in the first step, 0.45\% of E$\rightarrow$E pairings are false detections.
This is around 5 times larger than expected from \figureref{fig:erasuresweepM0}.
This difference stems from a series of sources.
Firstly, the judging upon if the two considered packets are in fact from the same meter is conservative.
Secondly, the effect of block errors as contrast to i.i.d. errors seems to play a role when observing the recorded data.
Thirdly, even though the receiver allows for $\frac{2}{32}$ bits to be in error in the synchronization word,
under these severe conditions, it is difficult for the receiver to retain synchronization.
Lastly, there exist repeating devices in the deployed network, forwarding a meters transmission to the end receiver.
This repetition does not follow the synchronous transmission scheme presented in this work,
and these repeated arrivals are likely to influence the performance number of the \ISM algorithm in a negative direction.

A final note to the results is the variation in probability of false detection for different steps.
The reason for this is from the mechanism in \ISM where a virtual slot is not setup for the next step, if a candidate is found for the current step.
This gives that the number of pairings for higher steps decreases, while the number of false detections largely remains the same,
and hence the relative ratio increases.

\begin{table}
\centering
\footnotesize
\subfloat[Overall experiemt values.]{%
\begin{tabular}{|lr|}\hline
Experiment duration & 628 minutes \\\hline
No. of different meters observed & 53 \\\hline
No. of erroneously received packets & 29794 \\\hline
No. of correctly received packets & 59683 \\\hline
\end{tabular}}\\
\subfloat[Packet pairings per \underline{C}orrect and \underline{E}rroneous. Last row is false positives.]{%
\bgroup
\setlength{\tabcolsep}{0.4em}
\begin{tabular}{|lrrrrrrrrrr|}\hline
Step & 1 & 2 & 3 & 4 & 5 & 6 & 7 & 8 & 9 & 10 \\\hline
C$\rightarrow$C & 43599 & 3627 & 573 & 121 & 63 & 37 & 19 & 26 & 8 & 2 \\
C$\rightarrow$E & 5036 & 962 & 264 & 112 & 42 & 33 & 11 & 38 & 8 & 4 \\
E$\rightarrow$C & 5047 & 916 & 278 & 87 & 42 & 32 & 17 & 29 & 11 & 4 \\
E$\rightarrow$E & 9166 & 3738 & 1879 & 1029 & 611 & 425 & 298 & 361 & 142 & 89 \\\hline
E$\rightarrow$E & 41 & 39 & 18 & 36 & 34 & 22 & 22 & 26 & 17 & 13 \\
f.d. \%         & 0.45 & 1.03 & 0.95 & 3.38 & 5.27 & 4.92 & 6.88 & 6.72 & 10.69 & 12.75 \\\hline
\end{tabular}\egroup}
\caption{Experment values.}
\label{tab:experiment}
\vskip-0.7cm
\end{table}

\section{Implementation Feasability}
\label{sec:feasability}
\begin{figure}
  \vskip-0.65cm
  \includegraphics[width=\columnwidth]{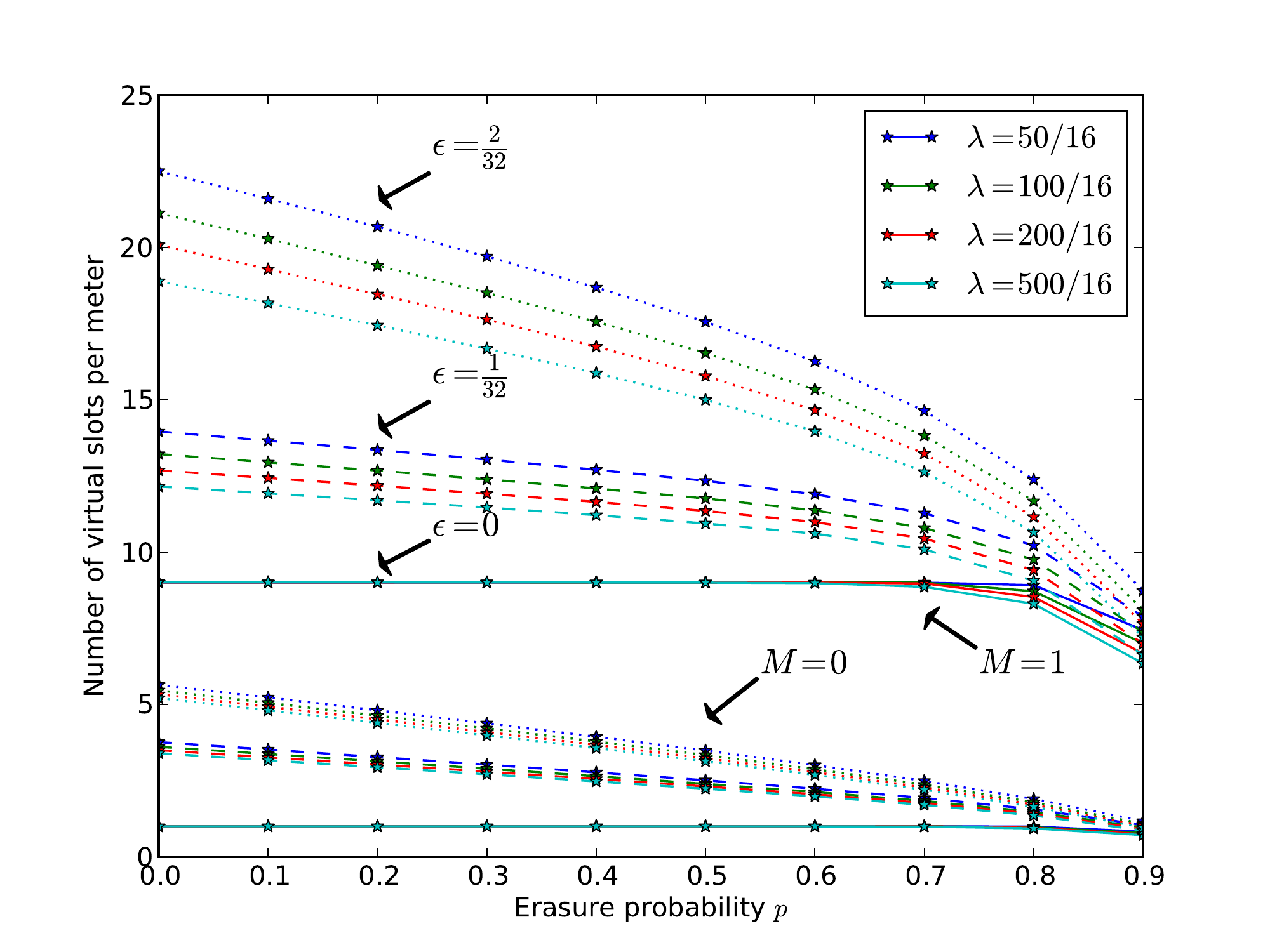}
  \caption{Average maximum number of simultaniously registered virtual slots per meter in the system.}
  \label{fig:memory}
  \vskip-0.7cm
\end{figure}

The overhead put on the receiver to perform the pairing should be fairly limited,
if further processing for packet recovery should be possible upon a pairing decision.
The virtual slots created when receiving an erroneous packet clearly imposes memory requirements,
but also the search for relevant virtual slots for an arriving packet can be a time consuming task.
The memory requirement implicitly gives an indication of the search time,
as the number of elements to search is proportional to the number of virtual slots.
In \figureref{fig:memory}, the simulated average maximum number of simultaneously registered virtual slots per number of participating meters is shown for various values of $M$, $\epsilon$, and $p$.
Clearly for $M=1$ the number of slots required is much greater than for $M=0$.
For $\epsilon=0$ and $p=0$ the values are as expected 1 or 9 virtual slots depending on $M$.
An increase of the bit error rate $\epsilon$ leads to a significant increase in the number of virtual slots,
because virtual slots created for an erroneously received \ACC will,
w.h.p., not match the next transmitted packet from the same meter.
As a consequence, the virtual slot is recomputed to find a candidate in the next window,
hence it will remain in the storage for a longer time (up to a maximum of 10 steps),
causing an increase in the number of virtual slots used per user.
For an increase of $p$ the number of virtual slots decreases,
because an erasure does not lead to virtual slots being created.

\section{Conclusion}
\label{sec:conclusion}
We show how a deterministic transmission interval between broadcast transmitted packets,
allows a receiver to derive the sender of the packet,
even though the packet payload itself is unreliable.
The focus of the analysis is the Wireless M-Bus protocol which provides strict timing in packet transmissions,
and we show that pairing validity depends on the number of meters within range of the receiver,
the number of tolerable bit errors, the bit error rate, and the packet erasure probability.

The proposed algorithm for packet pairing is a first step towards packet recovery across packet with the same meter data.
The experimental results from a deployed setup supports the use of a pairing mechanism for later recovery,
as more than $15\%$ of the conducted pairings on the receiver are between two erroneous packets,
leaving room for a substantial reception gain.

\bibliographystyle{IEEEtran}
\bibliography{literature}

\end{document}